\documentclass[sigconf]{acmart}
\renewcommand\footnotetextcopyrightpermission[1]{} 
\def\BibTeX{{\rm B\kern-.05em{\sc i\kern-.025em b}\kern-.08emT\kern-.1667em\lower.7ex\hbox{E}\kern-.125emX}}

\copyrightyear{2020}
\acmYear{2020}
\setcopyright{acmlicensed}
\acmConference[WSDM '20]{The Thirteenth ACM International Conference on Web Search and Data Mining}{February 5--9, 2020}{Houston, TX, USA}
\acmBooktitle{The Thirteenth ACM International Conference on Web Search and Data Mining (WSDM '20), February 5--9, 2020, Houston, TX, USA}
\acmPrice{15.00}
\acmDOI{DOI}
\acmISBN{ISBN}

\usepackage{array,multirow}
\usepackage{amsmath}
\usepackage[ruled, linesnumbered]{algorithm2e}
\usepackage{algpseudocode}
\usepackage[capitalize]{cleveref}
\usepackage[thinc]{esdiff}
\usepackage{enumitem}
\usepackage{caption}
\usepackage{graphicx,subcaption}
\usepackage{microtype}
\usepackage{pifont}
\usepackage{tikz}
\usetikzlibrary{arrows.meta}
\usepackage{todonotes}

\usepackage[aboveskip=1pt,skip=1pt,belowskip=0pt]{caption}

\newcommand{\hide}[1]{}

\newcommand{\xhdr}[1]{\vspace{0.5mm}\noindent{\textbf{#1.}}\hspace{0.5mm}}

\newcommand{\nbrweightsone}[1]{D(#1)}

\newcommand{\given}{\;\vert\;}
\newcommand{\prob}[1]{\textnormal{Pr}(#1)}

\newcommand{\inlinediff}[2]{\textnormal{d}#1 / \textnormal{d}#2}

\begin{document}

%
\title{Retrieving Top Weighted Triangles in Graphs}

\author{Raunak Kumar}
\authornote{Both authors contributed equally to this research.}
\affiliation{%
  \institution{Cornell University}
  \city{Ithaca}
  \state{NY}
}
\email{raunak@cs.cornell.edu}

\author{Paul Liu}
\authornotemark[1]
\affiliation{%
  \institution{Stanford University}
  \city{Stanford}
  \state{California}
}
\email{paul.liu@stanford.edu}

\author{Moses Charikar}
\affiliation{%
  \institution{Stanford University}
  \city{Stanford}
  \state{California}
}
\email{moses@cs.stanford.edu}

\author{Austin R. Benson}
\affiliation{%
  \institution{Cornell University}
  \city{Ithaca}
  \state{NY}
}
\email{arb@cs.cornell.edu}

%


\begin{abstract}
  Pattern counting in graphs is a fundamental primitive for many network
  analysis tasks, and a number of methods have been developed for scaling
  subgraph counting to large graphs. Many real-world networks carry a natural
  notion of strength of connection between nodes, which are often modeled by a
  weighted graph, but existing scalable graph algorithms for pattern mining are
  designed for unweighted graphs. Here, we develop a suite of deterministic and
  random sampling algorithms that enable the fast discovery of the
  3-cliques (triangles) with the largest weight in a graph, where weight is measured by a
  generalized mean of a triangle's edges. For example, one of our proposed
  algorithms can find the top-1000 weighted triangles of a weighted
  graph with billions of edges in thirty seconds on a commodity server,
  which is orders of magnitude faster than existing ``fast'' enumeration
  schemes. Our methods thus open the door towards scalable pattern mining in
  weighted graphs.
\end{abstract}

\maketitle

\section{Introduction}\label{sec:intro}

Small subgraph patterns, also called graphlets or network motifs, have proven
fundamental for the understanding of the structure of complex
networks~\cite{Milo-2004-superfamilies,Benson-2016-hoo,milo2002network}.
One of the simplest non-trivial subgraph patterns is the triangle (3-clique),
and the basic problem of triangle counting and enumeration has been studied
extensively from theoretical and practical
perspectives~\cite{Avron-2010-counting,Eden-2017-triangles,Seshadhri-2013-wedge,BFNPSW15,Stefani-2017-triest}.
These developments are often driven by the desire to scale graph counting to
large networks, where performing computations naively is intractable.
The focus on triangles is in part spurred by the widespread use of the pattern in
graph mining applications, including
community detection~\cite{berry2011tolerating,gleich2012vertex,Rohe-2013-blessing},
network comparison~\cite{contractor2006testing,mahadevan2007orbis,prvzulj2007biological},
representation learning~\cite{Henderson-2012-RolX,Rossi-2015-role}, and
generative modeling~\cite{Robins-2007-ERGMs,robles2016sampling}.
In addition, triangle-based network statistics such as the clustering
coefficient are used extensively in the social sciences
\cite{durak2012degree,lawrence2006organizational,burt2007secondhand,welles2010friend}.

Nearly all of the algorithmic literature on scalable counting or enumeration of
triangles focuses on \emph{unweighted} graphs.
However, many real-world network datasets have a natural notion of \emph{weight}
attached to the edges of the graph~\cite{barrat2004architecture}.
For example, edge weights can capture
tie strength in social networks~\cite{wasserman1994social},
traffic flows in transportation networks~\cite{Jia-2019-FlowSSL}, or
co-occurrence counts in projections of bipartite networks~\cite{Xu-2014-projection}.
Such edge weights offer additional insight into the structure of
these networks. Moreover, edge weights can enrich the types of small
subgraph patterns that are used in analysis. For instance, the network
clustering coefficient has been generalized to account for edge
weights~\cite{opsahl2009clustering,onnelaintensity}; in these cases, a triangle
is given a weight derived from the weights of its constituent edges. Roughly
speaking, the weight of a triangle is typically some combination of the
arithmetic mean, geometric mean, minimum and maximum of the edge weights of the
triangle.
All this being said, we still lack the algorithmic tools for fast analysis of
modern large-scale weighted networks, especially in the area of weighted
triangle listing and counting.

In applications of weighted triangles in this big data regime, it can often
suffice to retrieve only the $k$ triangles of largest weight for some suitable
$k$.
For example, in large online social networks, the weight of an edge could
reflect how likely it is for users to communicate with each other, and top
weighted triangles and cliques in this network could be used for group chat
recommendations. In such a scenario, we would typically only be interested in a
small number of triangles whose nodes are very likely to communicate with each
other as opposed to finding \emph{all} triangles in the graph.

Another application for finding top-weighted triangles appears in prediction
tasks involving higher-order network interactions. The goal of the
``higher-order link prediction'' problem is to predict which new groups of nodes
will simultaneously interact (such as which group of authors will co-author a
scientific paper in the future)~\cite{simplicial}. In this setting, existing
algorithms first create a weighted graph where an edge weight is the number of
prior interactions that involves the two end points and then predict that the
top-weighted triangles in this weighted graph will appear as higher-order
interactions in the future (weight here is measured by a generalized mean of the
triangle's edge weights).
Again, it is not necessary to find all triangles since only the top
predictions will be acted upon.
Existing triangle enumeration algorithms do not scale to massive graphs for these
problems, and we need efficient algorithms for retrieving triangles in large
weighted graphs.

%
In this work, we address the problem of enumerating the top-weighted triangles in a weighted graph.
To be precise, let $G=(V,E,w)$ be a simple, undirected graph with positive edge
weights $w$. We define the weight of a triangle in $G$ be equal to the
generalized $p$-mean; specifically, if a triangle $(i, j, k)$ has edge weights
$w_{ij}$, $w_{jk}$, and $w_{ik}$, then the triangle weight is
\begin{align}\label{eq:genmean}
m_p(i, j, k) := \left[\frac{1}{3}(w_{ij}^{p}+w_{jk}^{p}+w_{ik}^{p})\right]^{1/p}.
\end{align}
Given $G$ and an integer parameter $k$, we develop algorithms to extract the
\emph{top-$k$ triangles} in $G$. We use top-$k$ to refer to the triangles having
the $k$ largest weights, or in other words, the $k$-heaviest triangles. Note
that some special cases of the $p$-mean include arithmetic mean ($p=1$),
geometric mean ($p=0$), harmonic mean ($p=-1$), minimum ($p=-\infty)$ and
maximum ($p=\infty$). This family of means is more general and includes those
previously examined by Opsahl and Panzarasa \cite{opsahl2009clustering} and
Benson et al. \cite{simplicial}.
%

At a high level, we develop two families of algorithms for extracting
top-weighted triangles. The first family of algorithms
(Section~\ref{sec:heavylight}) is deterministic and optimized for extracting
top-$k$ weighted triangles for small $k$ (typically up to a few tens of
thousands). These algorithms take advantage of the inherent heavy-tailed edge
weight distribution common in real-world networks.
In the most general case, we show that under a modified
configuration model, these algorithms are even ``distribution-oblivious,'' in the
sense that they can automatically compute optimal hyper-parameters to the
algorithm for a wide range of input graph distributions. Additionally, the algorithmic analysis is done in a continuous sense (rather than discrete), which may be of independent interest.
The second family of algorithms (Section~\ref{sec:sampling}) is randomized and
aims to extract a large number of heavy triangles (not necessarily the
top-$k$). We show that this family of sampling algorithms is closely connected
to the prior sampling algorithms for \emph{counting} triangles on
\emph{unweighted} graphs~\cite{wedge}. Furthermore, we show that these sampling
algorithms are easily parallelizable.

We find that a carefully tuned parallel implementation of our deterministic
algorithm performs well across a broad range of large weighted graphs, even
outperforming the fast random sampling algorithms that are not guaranteed to
enumerate all of the top-weighted triangles. A parallel implementation of our
algorithm running on a commodity server with 64 cores can find the top 1000
weighted triangles in under 10 seconds on several graphs with hundreds of
millions of weighted edges and in 30 seconds on a graph with nearly two
billion weighted edges. We compare this with the off-the-shelf alternative
approach, which would be an intelligent triangle enumeration algorithm that
maintains a heap of the top-weighted triangles. Our proposed algorithms are
orders of magnitude faster than this standard approach.

\section{Additional related work}\label{sec:relatedwork}

Due its to the wide applicability, there is a plethora of work on
\emph{unweighted} triangle-related algorithms. In the context of enumeration algorithms, 
recent attention has focused on enumeration in the distributed and parallel
setting~\cite{CC11,SV11,AKM13,RH13}. These algorithms typically employ an
optimized brute force method on each machine~\cite{L08,BNPW10}, with the main
algorithmic difficulty in deciding how to partition the data amongst the
machines. Although each machine employs a brute force algorithm, early research
shows that these algorithms run in time almost linear in the number of edges so
long as the \emph{degeneracy} of the graph is small~\cite{CN85}, which has led
to efficient enumeration algorithms~\cite{BFNPSW15,SV11}. For comparison with our
methods, we modify such a fast enumeration method (specifically NodeIterator++~\cite{SV11}) to retain the top-$k$ weighted triangles. Although enumeration algorithms are agnostic to edge weights, we note that the sheer number of triangles in massive graphs renders such an approach prohibitively expensive.

When \emph{enumeration} becomes intractable, triangle-related algorithms focus
instead on merely triangle \emph{counts} or graph statistics such as clustering
coefficients. Again, these statistics are in the unweighted regime, as only the 
number of triangles are considered. There is a progression of sampling methods 
depending on what kind of structures one is sampling from the graph. At a basic level, edge-based
sampling methods sample an edge and counts all incident triangles on that edge.
So-called wedge-based methods sample length-2 paths~\cite{Seshadhri-2013-wedge},
and this concept has been generalized for counting 4-cliques~\cite{path}.
Finally, tiered-sampling combines sampling of arbitrary subgraphs to count the
occurrence of larger graphs (with a focus on 4-cliques and
5-cliques)~\cite{STU17}.

Beyond enumeration and sampling, there are numerous other methods for
triangle-based algorithms, such as graph sparsification
methods~\cite{tsourakakis2009doulion,PT12,ELT16}, spectral and matrix based
methods~\cite{T08,AYZ97}, and a multitude of methods for computing clustering
and closure coefficients~\cite{RH14,wedge,SW05-CC,YBL19}. For a deeper
background on triangle counting, we refer the reader to the overview of Hasan
and Dave~\cite{hasan18review}.

All of the above methods are for triangles. These ideas have been extended in
several ways. There are sampling methods for estimating counts of more general
motifs~\cite{Ahmed-2014-graphsample,JS17,Bressan-2017-graphlets} and motifs with
temporal structure~\cite{samplingtemporal} as well as parallel clique
enumeration methods~\cite{DBS18}. Still, these methods do not work for weighted
graphs, where subgraph patterns appear in generalizations of the clustering
coefficient~\cite{Saramaki-2005-weighted} as well as in graph
decompositions~\cite{Soufiani-2012-graphlet}.

\section{Deterministic algorithms}\label{sec:heavylight}

We begin by developing two types of deterministic algorithms for finding the
top-$k$ weighted triangles in a graph. As mentioned above,
the triangle weight is given by the generalized $p$-mean of its edge weights
as defined in \cref{eq:genmean}. A
simple and robust baseline algorithm is to employ a fast triangle enumeration algorithm for unweighted graphs, compute the weight on each triangle, and pick
out the top-$k$ weighted triangles according to the weighting function (or, to
save memory, we can maintain a heap of the top-weighted triangles). In our tests, we use an optimized sequential version of NodeIterator++~\cite{BFNPSW15, SV11, CN85}, which is the basis for many parallel enumeration algorithms.
We call this a ``brute force'' approach. We note that faster, parallel versions of the brute force approach exist. However, as the results in Section~\ref{sec:experiments} show, even a brute force enumerator with perfect parallelism would require over 2000 cores (or machines) to beat our sequential deterministic algorithm in certain cases.

The brute force approach is agnostic to the distribution of edge weights---it
is the same regardless of the weights. Intuitively however, we expect that triangles of large weight are 
formed by edges of large weight. We exploit this intuition below to develop
faster algorithms.
At a high level, our main deterministic algorithm will try to dynamically
partition the edges into ``heavy'' and ``light'' classes based on edge
weight. Following this partition, we find triangles adjacent to the heavy edges
until the top-$k$ heaviest are identified.

\begin{algorithm}[t]
  \DontPrintSemicolon  
  \SetKwInOut{Input}{Input}
  \SetKwInOut{Output}{Output}
  \KwIn{Weighted graph $G = (V,E, w)$, scaling $p$, number of triangles $k$}
  $H \gets \{e \in E : w_e > \tau \}$ \;
  $T \gets$ all triangles formed by edges in $H$ \;
  \KwRet{$k$ triangles in $T$ with largest $p$-mean weight.}    
  \caption{Static heavy-light algorithm for finding top-weighted triangles.}
  \label{alg:staticheavylight}
\end{algorithm}

\xhdr{A simple heavy-light algorithm}
As a precursor to our dynamic algorithm, consider a simple threshold based
algorithm as follows: (1) given a threshold $\tau$, partition the edges into a
``heavy'' set $H = \{e\,\mid w_e > \tau\}$ and a ``light'' set $L = E\setminus
H$.  For a large threshold $\tau$, we expect most of the edges to be in the
``light" class. Thus, the subgraph $G[H]$ induced by $H$ is small and we may run
any enumeration algorithm on it to get a collection of heavy triangles
(\cref{alg:staticheavylight}). This is not by itself guaranteed to find the heaviest triangles---edges in $H$ might only appear in triangles with edges in $L$. However, we find that in practice the heaviest triangles always have all of its edges in $H$. We note that with no additional \emph{asymptotic cost}, we can use existing triangle
enumeration algorithms to check for triangles with only one or two heavy edges
(and the rest light). Unfortunately, the constant factor slow downs substantially increase the running time on real-world graphs. 

In practice, we find that this simple algorithm vastly outperforms brute
force, and can always find the top-weight triangle given a proper threshold; 
thus, this method will serve as a robust baseline. 
Nonetheless, there are a couple of issues with this
static heavy-light algorithm. The first is that since the algorithm relies on a
single threshold $\tau$ to partition the edges into the light and heavy sets,
many more triangles may be enumerated than is necessary. The second is that it
is difficult to know what an appropriate value for $\tau$ may be given no prior
knowledge of the input graph. In the next section, we develop a dynamic variant
of Algorithm~\ref{alg:staticheavylight} to deal with these issues.

\subsection{Dynamic heavy-light algorithm}\label{subsec:heavylightalg}
\begin{algorithm}[t]
  \DontPrintSemicolon    
  \SetKwInOut{Input}{Input}
  \SetKwInOut{Output}{Output}
  \KwIn{Weighted graph $G = (V, E, w)$, scaling $p$, number of triangles $k$,
    parameter $\alpha_p$.}
  Sort $E$ in decreasing order of weight \;
  Set threshold $\tau = \infty$, triangle set $T = \emptyset$ \;
  Set partitions $S = H = \emptyset$, $L = E$ \;
  Set edge pointers $h = l = -1$ \;
  \tcp{We take the convention that $e_{-1} = \infty$.}
  \While{there are $< k$ triangles above weight $\tau$ in $T$}{\label{algline:dynamicwhileloop}
    \uIf{$w_{l+1} > w_{h+1}^{\alpha_p}$}{\label{algline:dynamiclabelltoh}
      Move $e_{l+1}$ from $L$ to $H$. \;
      $Y = $ triangles formed by $e_{l+1}$ and 2 edges from $S \cup H$ \;
      $Z = $ triangles formed by $e_{l+1}$, 1 edge from $L$ and 1 edge from $S \cup H$ \;
      $T = T \cup (Y \cup Z)$ \;
      $l = l + 1$ \;
    } \Else{\label{algline:dynamiclabelhtos}
      Move $e_{h+1}$ from $H$ to $S$. \;
      $Y = $ triangles formed by $e_h$ and 2 edges from $L$ \;
      $T = T \cup Y$ \;
      $h = h + 1$ \;
    }
    Update threshold $\tau = w_h^p+2w_l^p$ \;\label{algline:dynamicthresholdset}
  }
  \KwRet{$k$ triangles in $T$ with largest $p$-mean weight.}  
  \caption{Dynamic heavy-light algorithm for finding the top-$k$ weighted triangles.}
  \label{alg:dynamicheavylight}
\end{algorithm}

We now develop a dynamic algorithm that uses the concepts of
\cref{alg:staticheavylight} but is significantly more efficient. 

Suppose we have preprocessed the edges $E = \{e_0, e_2, ..., e_{m-1}\}$ so that they
are sorted by \emph{decreasing} weight ($w_i \geq w_{i+1}$ where $w_i$ denotes the weight of the $i$-th edge). Our dynamic
heavy-light algorithm maintains a dynamic partition of the edges into
\emph{three} sets based on weight:
\begin{itemize}[topsep=0pt]
\item $S = \{e_0, \ldots, e_{h}\}$ are the ``super-heavy'' edges
  of the $h$ largest weights;
\item $H = \{e_{h+1}, \ldots, e_{l}\}$ are the ``heavy'' edges consisting of the next $l-h$ edges of largest weights; and
\item $L = \{e_{l+1}, \ldots, e_{m-1}\}$ are the remaining ``light'' edges
  that are neither ``heavy'' nor ``super-heavy.''
\end{itemize}
As the algorithm evolves, we adjust the sets $S$, $H$, and $L$ by changing the values of $h$ and $l$. Any triangle can
be broken down into a combination of ``super-heavy,'' ``heavy,'' and ``light''
edges. As a first order approximation, we intuitively expect the heaviest class
of triangles to have three super-heavy edges, the next heaviest to have two
super-heavy edges and one heavy edge, and so on down to the case of three light edges. Furthermore, by considering the
edges in a specific order, we can also obtain useful bounds on the weight of the
heaviest triangles we have not yet enumerated. Suppose we have
enumerated all triangles containing three super-heavy edges. Then the heaviest
triangle not yet enumerated must have at least 1 edge from either $H$ or $L$. This upper
bounds the $p$-th power of the weight of that triangle to be $\frac{1}{3}\left(2w_0^p + w_h^p\right)$. Our method will
try to enumerate triangles so that this bound decreases as quickly as possible.

\Cref{alg:dynamicheavylight} outlines our procedure.
Each step of the algorithm consists of two steps:
(i) update the partition by moving an edge to a heavier class; and
(ii) enumerate triangles whose edges come from certain classes.
At the end of each step, we maintain the invariant that we have enumerated all
triangles with at least one super-heavy edge or at least two heavy-edges. This
invariant is what allows us to obtain a bound $\tau$ on the heaviest triangle we have
not yet enumerated. The constant $\alpha_p$ of 
\cref{alg:dynamicheavylight},
determines how edges get promoted to heavier classes; this parameter is used in
our analysis to optimize the expected decrease in the threshold $\tau$. We will specify this constant later in our analysis. When
the algorithm begins, the partitions are initialized with $S=H=\emptyset$,
$L=\{e_0,\ldots,e_{m-1}\}$, and $\tau = \infty$.  The algorithm runs until it
enumerates $k$ triangles above a dynamically decreasing threshold $\tau = w_h^p
+ 2w_l^p$ (\cref{algline:dynamicthresholdset} of
\cref{alg:dynamicheavylight}).

Let $\tau^\ast$ be the weight of the $k$-th heaviest triangle. As soon as $\tau
\leq \tau^\ast$, we will have enumerated all of the top-$k$ triangles. This
algorithm solves both issues of our static heavy-light algorithm. If the
threshold $\tau$ hits $\tau^\ast$ exactly, we only enumerate around $k$
triangles. As $\tau$ is computed on the fly, there is no need
to choose the threshold at which we partition the edges. 
As a further benefit, we give a parameter-free version of our algorithm where $\alpha_p$ is \emph{implicitly} computed. The algorithmic analysis is done in a continuous sense (rather than discrete), which may be of independent interest.

\subsection{Algorithm correctness}
\label{sec:correctness}
We first bound the heaviest possible triangle not yet enumerated.
Observe that when an edge moves from $L$ to $H$ (\cref{algline:dynamiclabelltoh}) all triangles
including that edge and at least one edge from $S \cup H$ are enumerated; when an
edge moves from $H$ to $S$ (\cref{algline:dynamiclabelhtos}) all triangles including
that edge are enumerated.
Thus, the if-else clause maintains the invariant that all triangles with at least one edge
from $S$ or at least two edges from $H$ are enumerated.
Now consider any triangle $\Delta$ with weight at least $w_h^p + 2w_l^p$. By case analysis, there must either exist two edges with weight at least $w_l$, or one edge with weight at least $w_h$. This means that either two edges are from $H$, or one edge is from $S$. In either case, our invariant ensures that $\Delta$ must have been enumerated.
In fact, similar reasoning shows a tight threshold is $w_{h+1} + w_{l+1} + w_{l+2}$, as $e_{h+1}, e_{l+1}, e_{l+2}$ is potentially an unenumerated triangle. However, this sum is at most $w_h + 2w_l$ due to the monotonicity of the edge weights. 

Since $\tau$ is monotonically decreasing, this implies that at
the end of each iteration of the while loop in \cref{alg:dynamicheavylight}, $T$
contains the top-$\lvert T \rvert$ weighted triangles in the graph (there may be triangles not enumerated with \emph{equal} weight to one of the
triangles in $T$). Therefore, \cref{alg:dynamicheavylight} correctly returns the
top-$k$ triangles, provided the graph has at least $k$ triangles.
  

\subsection{Algorithm analysis}
\label{sec:analysis}
Although \Cref{alg:dynamicheavylight} is a discrete algorithm we present a simple analysis using continuous differentials.  Let $w_h(t)$ and
$w_l(t)$ denote the weight of edges $e_h$ and $e_l$ at time $t$
respectively (one can think of $t$ as a continuous counter for the while loop iterations).  At time $t$, the threshold is $\tau(t) = w_h(t)^p +
2w_l(t)^p$.  As a proxy to optimizing the triangles enumeration rate, we would like to maximize the rate at which the threshold decreases instead.  To do this, we maximize the derivative $\inlinediff{\tau}{t}$ by adjusting $w_h(t)$ and $w_l(t)$ based on the input parameter $\alpha_p$.  

Let us consider what $\inlinediff{w_h}{t}$ and $\inlinediff{w_l}{t}$ represent. The derivatives approximate the maximum change we can make in $w_h$ or $w_l$ ``per unit of computation time". In each iteration of the while loop, we can choose to spend time decreasing $w_h$ or $w_l$. Thus, a rough approximation to the derivatives is the ratio of the change in weight (by incrementing either the $h$ or $l$ pointer) to the computational cost of changing the corresponding pointer. 

Let $w_{<h} := \max \{w : w < w_h\}$ and let $CDF(w)$ and $PDF(w)$ be the cumulative density function and probability density function of the edge weight distribution respectively. If we move the edge pointer $e_h$, then the average change in $w_h$ is the ratio of $(w_h - w_{<h})$ to the number of edges that have
weight $w_h$.  The number of edges that have weight $w_h$ is proportional to $CDF(w_h) - CDF(w_{<h})$, so the average change in $w_h$ is approximately 
\begin{equation}
    \label{eq:avgchange}
    \frac{w_h - w_{<h}}{CDF(w_h) - CDF(w_{<h})} \approx \frac{1}{PDF(w_h)}.
\end{equation} 
Similarly, the change in $w_l$ is equal to $1/\text{PDF}(w_l)$.

\xhdr{Analysis for power-law distributed weights}
At this point, we are free to continue our analysis with any model for for the
distribution of weights on $G$. One important case which is analyzable is a
power law distribution on the edges, and we find that this type of distribution
is a reasonable model for several of our datasets (see
\cref{fig:edge_weights_plot} for examples). Thus in this section, we carry out
the analysis assuming that the edge weights follow a power law distribution with
parameter $\beta$.

\begin{figure}[tb]
  \centering
  \includegraphics[width=0.495\columnwidth]{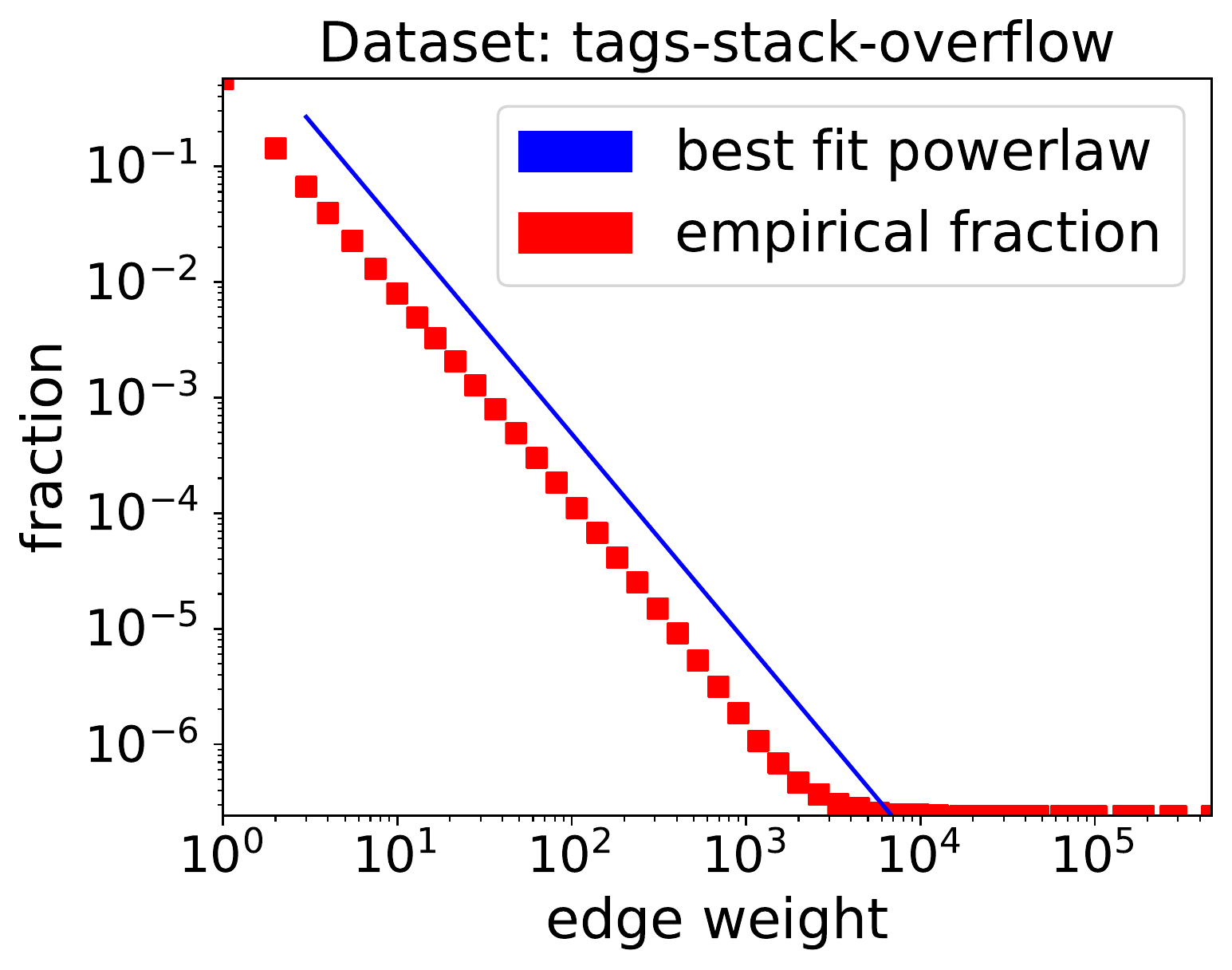}
  \includegraphics[width=0.495\columnwidth]{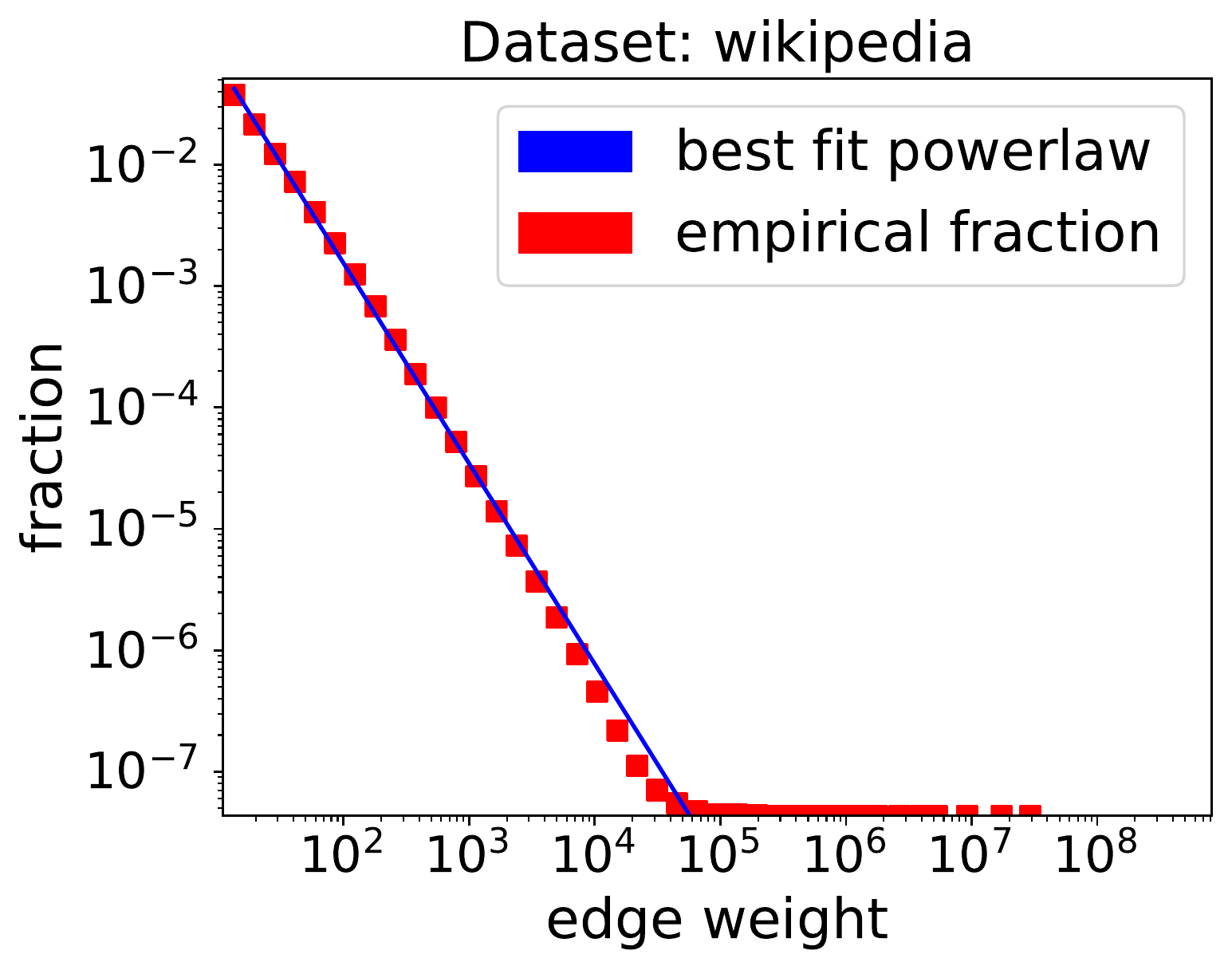}
  \caption{Edge weight distribution in two datasets (see
    \cref{sec:data} for a description of the datasets). These plots suggest that
    a power law distribution on edge weights is a reasonable assumption. With
    this, we have a simple condition (\cref{algline:dynamiclabelltoh}) to choose
    which pointer to move in the dynamic heavy-light algorithm
    (\cref{alg:dynamicheavylight}).}
\label{fig:edge_weights_plot}
\end{figure}

Formally, let $X$ be a random variable. We say $X$ follows a power law
distribution with parameter $\beta>1$ and some constant $a>0$ if $P(X \geq x) \sim a x^{1-\beta}$ for large $x$. Thus, the probability that a random edge
weight is greater than or equal to $w$ is $O(w^{1-\beta})$ and this implies that
the probability that a random edge weight is equal to $w$ is $O(w^{-\beta})$ for
large $w$.  Using this assumption we can write the change in $w_{h}$ and
$w_{l}$ as $O(w_{h}^\beta)$ and $O(w_{l}^\beta)$ respectively.

Now we analyze the computational cost of changing $e_h$ and $e_l$. To do this we impose a simple configuration model on the way that $G$ is generated~\cite{newman2001random}. We assume that each vertex $v$ draws its degree $d$ from a univariate degree distribution $\mathcal{D}$ with the sum of degrees being even. We assume that the graph is generated from the following random process: (1) each vertex $v$ starts out $d_v$ stubs. While there are stubs available, two random stubs are drawn from the set of all stubs, and the vertices corresponding to those stubs are connected. Furthermore, upon connection a random edge weight drawn from the edge weight distribution is assigned to the edge. At the end, all self-loops in the graph are discarded. Note that these assumptions are quite strong, however we find that even this simple analysis yields good estimates for optimal values of $\alpha_p$ in practice (see \cref{sec:experiments}). 
Continuing with the analysis, we now analyze the expected cost in incrementing either the $h$ or $l$ pointer. Let $G_H := G[S \cup H]$ and $G_L := G[L]$, and let $\bar{d}_U$ denote the average degree in a graph $U$. With appropriate data structures for checking the existence of edges in $G_H$, the cost of incrementing $h$ is bounded by the degree sum of the endpoints of $e_h$ in $G_L$, which is on average $O\left(\bar{d}_{G_L}\right)$. Assuming that $G_L$ has approximately as many edges as $G$ (valid in the case of small $k$), $\bar{d}_{G_L} \approx \bar{d}_{G}$. Thus the computational cost of moving $e_h$ is approximately $O\left(\bar{d}_{G}\right)$.

Similarly, the computational cost of moving $e_l$ is bounded by the degree sum of the endpoints of $e_l$ in $G_H$, which is on average $O\left(\bar{d}_{G_H}\right)$. Since the number of edges in $G_H$ is exactly $|S\cup H|$, the assumptions on the weight distribution gives us that $\bar{d}_{G_H} = O\left( CDF(w_l) \bar{d}_G \right)$. Thus the computational cost of moving $e_l$ is approximately $O\left(w_{l}^{1-\beta} \bar{d}_{G}\right)$.

Combining this with \cref{eq:avgchange}, we obtain the following expressions for the derivatives
\begin{equation}\label{eq:derivative}
    \inlinediff{w_{h}}{t} = O\left(w_{h}^\beta / \bar{d}_{\text{G}}\right), \quad
    \inlinediff{w_{l}}{t} = O\left(w_{l}^{2\beta-1} / \bar{d}_G\right).
\end{equation}

\hide{ In each iteration of the algorithm we can only change one of $e_h$ or
  $e_l$, so which one should we choose if our goal is to maximize the rate at
  which the threshold decreases?  }

Note that since $w_l$ and $w_h$ are decreasing, both derivatives are
monotonically decreasing as the algorithm progresses. This monotonicity property
means that greedily choosing the pointer to increment is optimal. The threshold
decrease rate is
$\inlinediff{\tau}{t} = p w_h^{p-1} \inlinediff{w_h}{t} + 2p w_l^{p-1} \inlinediff{w_l}{t}$.
Since at each iteration of the algorithm we can only choose to change one of
$e_h$ or $e_l$, we should greedily change the pointer that gives the most ``bang
per buck'', i.e., choose $e_h$ and $e_l$ such that
\begin{equation}
    \label{eq:ptr_values}
    w_h^{p-1} \inlinediff{w_h}{t} = 2w_l^{p-1} \inlinediff{w_l}{t}
    \implies w_{h} \sim O\left( w_{l}^{2-\frac{p}{p-1+\beta}} \right).
\end{equation}
In other words, we should maintain the edge pointers $e_h$ and $e_l$ such that
the weights are separated geometrically by $\alpha_p = 2-\frac{p}{p-1+\beta}$.
  
\xhdr{Distribution oblivious dynamic heavy-light algorithm}
The analysis in the previous section yields a fast algorithm given a known prior
on the power law parameter of the weight distribution. In many applications,
this can be easily and robustly estimated. In this section we present a method
for which the parameter $\alpha_p$ can be implicitly estimated on the fly,
albeit with less robustness in practice than simply setting $\alpha_p$.

Although we assumed a power-law distribution on the edge weights, our analysis
is actually much more general than that. Note that as long as the derivatives
$\inlinediff{w_{h}}{t}$ and $\inlinediff{w_{l}}{t}$ are monotonic, our greedy
method of incrementing the pointers will be optimal. For the derivatives to be
monotonically decreasing, the only requirement is that the $PDF$ of the weight
distribution is monotonically increasing as the weight decreases. This includes
a wide family of distributions such as power laws and uniform distributions.

Furthermore, the analysis we used to derive $\alpha_p$ can also be used to
compute $\alpha_p$ implicitly. By maintaining an estimate of the derivatives
$\inlinediff{w_{e_h}}{t}$ and $\inlinediff{w_{e_l}}{t}$ as the algorithm runs,
we can compute all the derivatives used in the analysis on the fly and greedily
change the pointer with higher value of $w^{p-1} \inlinediff{w}{t}$
(Equation~\ref{eq:ptr_values}).

Following the analysis in the previous section, the change in weight for $e_h$
is estimated by the ratio of $w_{h} - w_{<h}$ and the number of edges that have
weight $w_{h}$, and similarly for $e_l$. The computational cost of changing $e_l$
can be estimated by the sum of the degrees of the endpoints of $e_l$ in $G_H$,
and similarly for $e_h$ with $G_L$. Consequently, we can obtain a
``distribution-oblivious'' algorithm that works on a family of monotone
distributions.

In our experiments we find that this automatic way of implicitly computing
$\alpha_p$ is quite successful, although in practice noise in the derivative
estimates may cause this algorithm to be slower than using a set value of
$\alpha_p$.

\section{Random sampling algorithms}\label{sec:sampling}

In this section, we develop random sampling algorithms designed to sample a large collection of triangles with large weight. More formally, given a generalized $p$-mean as a
weight function, these algorithms all sample triangles exactly proportional to
their weight. The main difference between the algorithms is how efficiently they can
generate samples.

We specifically generalize three types of sampling schemes that have been used to
estimate triangle counts in unweighted graphs. The first scheme is based on
\emph{edge sampling}~\cite{kolountzakis2012efficient,tsourakakis2011triangle,pagh2012colorful},
where we first sample one edge and then enumerate triangles adjacent to the sampled edge.  The second
method uses ideas from \emph{wedge sampling}~\cite{wedge}, where we sample two adjacent
edges and check whether
these two edges induce a triangle. The final approach generalizes the idea of path sampling~\cite{path}; we sample a three-edge path and check if it induces a triangle.
Although these approaches were are all
designed for triangle counting in unweighted graphs, they generalize quite
seamlessly to a simple scheme for sampling highly weighted triangles. The main
benefits of these algorithms are that they are simple to implement and also easy
parallelize, since samples can be trivially generated in parallel.

Throughout this section, we assume that our weighting function for a triangle is
a generalized $p$-mean as given in \cref{eq:genmean}. Since the weighted
ordering of triangles is independent of the scaling by $1/3$ and the exponent
$1/p$, we can instead consider the more simple weighting function
\begin{align}\label{eq:pmean_unnorm}
w_p(a, b, c) = w_{ab}^p + w_{bc}^p + w_{ac}^p.
\end{align}
For a given vertex $a \in V$, we will use $N(a)$
to denote the set of neighbors of $a$:
$N(a) = \{b \in V \given (a, b) \in E\}$,
and $d_a$ to be the (unweighted) degree of node $a$ (i.e., $d_a = \lvert N(a) \rvert$).

\subsection{Weighted edge sampling}\label{subsec:samplingalg}
\begin{algorithm}[tb]
  \DontPrintSemicolon
  \SetKwInOut{Input}{Input}
  \SetKwInOut{Output}{Output}
  \KwIn{Weighted graph $G = (V,E,w)$, scaling $p$, number of iterations $t$,
  number of triangles $k$}
  Initialize triangle set $T \gets \emptyset$ \;    
  \For{iteration $1, \ldots, t$}{
    Sample edge $(a, b) \propto w_{ab}^p$\;
    \For{each neighbor $c \in N(a) \cap N(b)$}{
      $T \gets T \cup \{(a,b,c)\}$\;
    }
  }
  \KwRet{$k$ triangles in $T$ with largest $p$-mean weight.}
  \caption{Weighted edge random sampling algorithm}
  \label{alg:edge}
\end{algorithm}

We first discuss the edge sampling algorithm (\cref{alg:edge}). The algorithm is
based on a simple two-step procedure. First, we sample a single edge according to
the following distribution
\[
\textstyle \prob{\text{sample edge $(a, b)$}} := w_{ab}^p / Z,\; Z = \sum_{(u, v) \in E} w_{uv}^p.
\]
Second, after we sample an edge $(a, b)$, we enumerate all triangles $(a, b, c)$
incident to $(a, b)$. These two steps are repeated several times.

The above procedure has a few issues. First, it takes $O(d_a + d_b)$ time to
find triangles adjacent to an edge $(a, b)$, which can be expensive in graphs
where high-degree nodes are connected; we get around this in our next sampling
scheme.  Second, there is no guarantee that the above procedure will generate at least $k$ unique triangles, and moreover, even if the algorithm samples a
sufficient number of triangles, it is not necessarily the case that these are
the \emph{top-weighted} triangles. The latter issue is an inherent limitation of
random sampling schemes in general. All that being said, the procedure has the
nice property of being biased in terms of sampling triangles with high weight,
formalized as follows.
\begin{proposition}\label{prop:edge}
  The probability that a triangle $(a, b, c)$ is enumerated in a given
  iteration of \cref{alg:edge} is $w_p(a, b, c) / Z$, where $Z = \sum_{e \in E} w_e^p$.
\end{proposition}
\begin{proof}
   The probability that any edge $(u,v)$ is sampled initially is $w_{u,v}^p / Z$.
   Triangle $(a,b,c)$ is enumerated if any one of if edge $(a, b)$, $(b, c)$, or $(a, c)$
   is this sampled edge.
\end{proof}

While the ES method is simple to describe, making the algorithm fast in practice
requires some careful implementations. First, a natural way of sampling an edge
is to simply pick one at random with probability proportional to its weight, but
this is slow because there are a large number of edges. However, there is
typically a much smaller number of \emph{unique} edge weights. Thus, we first
first sample an edge weight and then sample an edge with this weight.
Pre-processing calculation of sampling probabilities for this approach involves
iterating over the edges and computing two quantities---a cumulative edge weight
(in order to sample an edge weight) and a map of edge weight to edges (in order
to sample an edge given an edge weight). This pre-processing step can take much
longer than the sampling loop if implemented naively. However, in a sorted list
of edges, all edges that share the same edge weight lie in a contiguous chunk
and this significantly speeds up the process of computing the above
quantities.


\subsection{Weighted wedge sampling}
\begin{algorithm}[tb]
  \DontPrintSemicolon
  \SetKwInOut{Input}{Input}
  \SetKwInOut{Output}{Output}
  \KwIn{Weighted graph $G = (V,E,w)$, scaling $p$, number of iterations $t$,
    number of triangles $k$}
  Initialize triangle set $T \gets \emptyset$ \;  
  \For{iteration $1, \ldots, t$}{
    Sample node $a$ with probability as in \cref{eq:wedgeone} \;
    Sample $b \in N(a)$ with probability as in \cref{eq:wedgetwo} \;
    Sample $c \in N(a)$ with probability as in \cref{eq:wedgethree} \;
    \uIf{nodes $a,b,c$ form a triangle}{
      $T \gets T \cup \{(a,b,c)\}$\;
    }
  }
  \KwRet{$k$ triangles in $T$ with largest $p$-mean weight.}
  \caption{Weighted wedge random sampling algorithm}
  \label{alg:wedge}
\end{algorithm}

One of the issues with the simple edge sampling scheme described above is that
we have to look over the neighbors of the end points of the sampled edge in
order to find triangles. This can be expensive if the degrees of these nodes are
large. An alternative approach is to sample adjacent edges---also called
\emph{wedges}---with large weight and then check if each wedge induces a
triangle. This scheme is called \emph{wedge sampling} and has been
used as a mechanism for estimating the total number of triangles in an
unweighted graph~\cite{wedge,TurkogluT17}.

The overall sampling scheme is outlined in \cref{alg:wedge}.  Each iteration of
\cref{alg:wedge} proceeds in three steps. First, we sample a single node with a
bias towards nodes that participate in heavily weighted edges. Specifically, let
$\nbrweightsone{a} = \sum_{b \in N(a)} w_{ab}^p$ denote the sum of the edge weights incident to $a$.
We sample a node $a$ according to the following distribution:
\begin{align}
\displaystyle \prob{\text{sample $a$}} = \tilde{W}_{1}(a) / Z_1 := 2 \cdot d_a \cdot \nbrweightsone{a} / Z_1, \label{eq:wedgeone}
\end{align}
where $Z_1$ is a normalizing constant. Next, we sample a neighbor of node $a$,
again with a bias towards nodes that participate in heavily weighted edges. The
specific distribution is
\begin{equation}
\begin{split}
   \prob{\text{sample $b \in N(a) \given a$}} & = \tilde{W}_{2}(b \given a) / Z_2 \\
                                              & := (d_a \cdot w_{ab}^p + \nbrweightsone{a}) / Z_2, \label{eq:wedgetwo}
\end{split}
\end{equation}
where $Z_2$ is a normalizing constant.
We have now produced a single edge and want to produce an adjacent edge.
We do this by sampling another neighbor of $a$, this time with probability
\begin{equation}
\begin{split}
    \prob{\text{sample $c \in N(a) \given a,b$}} & = \tilde{W}_{3}(c \given a,b) / Z_3 \\
                                                 & := (w_{ac}^p + w_{ab}^p) / Z_3, \label{eq:wedgethree}
\end{split}
\end{equation}
where $Z_3$ is again a normalizing constant. If the sampled wedge
$\{(a,b), (a,c)\}$ induces a triangle, then we add it to our collection.

Similar to the unweighted scheme of Seshadhri et al.~\cite{wedge}, we show the following:
\begin{proposition}\label{prop:wedge}
  A given iteration of \Cref{alg:wedge} samples triangle $(a, b, c)$ with
  probability $w_p(a,b,c) / \tilde{Z}$, where $\tilde{Z} = \sum_{v \in V} d_v \cdot D(v)$.
\end{proposition}
\begin{proof}
The normalizing constants in \cref{eq:wedgeone,eq:wedgetwo,eq:wedgethree} are
  $Z_1 = 2 \sum_{v \in V} d_v \cdot D(v)$,
  $Z_2 = \tilde{W_{1}}(a)$, and
  $Z_3 = \tilde{W}_2(b \given a)$
  Thus, the probability of sampling wedge $(a,b,c)$ centered on $a$ is equal to
  \[
  \frac{\tilde{W}_1(a)}{Z_1} \cdot \frac{\tilde{W}_2(b \given a)}{Z_2} \cdot \frac{\tilde{W}_3}{Z_3} = \frac{\tilde{W}_3}{Z_1} = \frac{w_{ab}^p + w_{ac}^p}{Z_1}.
  \]
  Thus, the probability of sampling any of the three wedges consisting of nodes $a,b$ and $c$ is equal to
  \[
  (w_{ab}^p + w_{ac}^p) / Z_1 + (w_{ac}^p + w_{bc}^p) / Z_1 + (w_{bc}^p + w_{ab}^p) / Z_1.
  \]
  Plugging in the expression for $Z_1$ shows that this probability is equal to $w_p(a,b,c) / \tilde{Z}$.
  Since triangle $(a,b,c)$ is sampled if any of its three wedges is sampled, this yields the desired result. 
\end{proof}


\subsection{Weighted path sampling}
\begin{algorithm}[t]
 \DontPrintSemicolon
 \SetKwInOut{Input}{Input}
 \SetKwInOut{Output}{Output}
 \KwIn{Weighted graph $G = (V,E,w)$, scaling $p$, number of iterations $t$, number of triangles $k$.}
 Initialize triangle set $T = \emptyset$ \;
 \For{$iteration \, 1, \ldots, t$}{
    Sample edge $e = (a,b)$ with probability as in \cref{eq:pathone} \;
    Sample $c \in N(a) \setminus b$ with probability as in \cref{eq:pathtwo} \;
    Sample $c' \in N(b) \setminus a$ with probability as in \cref{eq:paththree} \;
    \uIf{$c = c'$}{
        $T \gets T\cup\{(a,b,c)\}$ \;
    }
 }
 \KwRet{$k$ triangles in $T$ with largest $p$-mean weight.}
 \caption{Weighted path sampling algorithm}
 \label{alg:path}
\end{algorithm}
We can also design a sampling scheme based on path sampling~\cite{path}.
Previously, edge sampling found triangles after sampling a single edge and
wedge sampling found triangles after sampling two pairs; here,
path sampling starts by sampling three edges that will be biased towards finding
a weighted triangle.
This sampling scheme performs poorly in practice, but we include it here for the sake of theoretical interest.

For path sampling, we first sample an edge $(a, b)$ and then sample two more edges $(c,a)$ and $(b, c')$ connected to $(a,b)$ (producing a length-three path). A triangle is obtained when this length three path closes into a triangle ($c = c'$). 
Let $\tilde{d}_v := d_v - 1$ and $\tilde{D}_u(v) := D(v) - w_{uv}^p$. The sampling probabilities are as follows:
\begin{multline}
   \prob{\text{sample $(a, b)$}} = \tilde{W}_1(a,b) / Z_1 \\
                                  := \left(\tilde{d}_a\tilde{d}_b w_{ab}^p + \tilde{d}_a\tilde{D}_a(b) + \tilde{d}_b \tilde{D}_b(a)\right) / Z_1, \label{eq:pathone}
\end{multline}
\begin{multline}
   \prob{\text{sample }c\in N(a)\setminus{b} \given a,b} = \tilde{W}_2(c\given a,b) / Z_2 \\
                                                          := \left(\tilde{d}_b \cdot (w_{ac}^p + w_{ab}^p) + \tilde{D}_a(b)\right)/Z_2, \label{eq:pathtwo}
\end{multline}
\begin{multline}
   \prob{\text{sample }c'\in N(b)\setminus{a} \given a,b} = \tilde{W}_3(c'\given a,b) / Z_3 \\
                                                           := \left(w_{ac}^p + w_{ab}^p + w_{bc'}^p\right)/Z_3, \label{eq:paththree}
\end{multline}
where $Z_1, Z_2, Z_3$ are normalizing constants.

Under these sampling probabilities, the path sampling method (\cref{alg:path})
samples triangles proportional to their weight, just as in
\cref{alg:edge,alg:wedge}.
\begin{proposition}\label{prop:path}
  A given iteration of \Cref{alg:path} samples triangle $(a, b, c)$ with
  probability $w_p(a,b,c) / \tilde{Z}$, where \\ $\tilde{Z} = \sum_{v \in V} \left(\tilde{d}_a\tilde{d}_b w_{ab}^p + \tilde{d}_a\tilde{D}_a(b) + \tilde{d}_b \tilde{D}_b(a)\right)$.
\end{proposition}
\begin{proof}
 Let $Z_1, Z_2$ and $Z_3$ be the normalization constants of $\tilde{W}_1(a,b),
\tilde{W}_2(c\given a,b)$ and $\tilde{W}_3(c'\given a,b)$ (equations
\eqref{eq:pathone}, \eqref{eq:pathtwo} and \eqref{eq:paththree})
respectively. We have
\begin{align*}
    Z_2 &= \sum\limits_{c \in N(a) \setminus b} \tilde{d}_b(w_{ac}^p + w_{ab}^p) + \tilde{D}_a(b) \\
        &= \tilde{d}_b \tilde{D}_b(a) + \tilde{d}_a\tilde{d}_b w_ab + \tilde{d}_a \tilde{D}_a(b) \\
        &= \tilde{W}_1(a,b). \\
    Z_3 &= \sum\limits_{c' \in N(b) \setminus a} w_{ac}^p + w_{ab}^p + w_{bc'}^p \\
        &= \tilde{d}_b \cdot (w_{ac} + w_{ab}) + \tilde{D}_a(b) \\
        &= \tilde{W}_2(c\given a,b).
\end{align*}
The probability of sampling triangle $(a,b,c)$ is equal to
$$
\frac{\tilde{W}_1(a,b)}{Z_1} \cdot \frac{\tilde{W}_2(c\given a,b)}{Z_2} \cdot \frac{\tilde{W}_3(c'\given a,b)}{Z_3} = \frac{\tilde{W}_3(c'\given a,b)}{Z_1} = \frac{w_p(a,b,c)}{Z_1}.
$$
Thus, the probability of sampling triangle $(a,b,c)$ is
proportional to its weight.
\end{proof}

\subsection{Number of samples}

\Cref{prop:edge,prop:wedge,prop:path} says that \cref{alg:edge,alg:wedge,alg:path} tend to sample
edges with large weight.
However, there is no guarantee that the top-weighted triangles are actually
enumerated.
Some standard probabilistic analysis can at least give us some sense on how many
iterations we need.
Specifically, if $q = w_p(a,b,c) / Z$ is the probability of sampling a triangle
$(a,b,c)$, then for any $\delta \in (0,1)$, $s \geq \log(1/\delta) / q$ samples
guarantees that $(a,b,c)$ is enumerated with probability at least $1 - \delta$.
To see this, let $r = 1 - (1-q)^s$ denote the probability that triangle
$(a,b,c)$ is sampled at least once in $s$ samples. Using the fact that $1-x \leq \exp(-x)$,
we get that $r \geq 1 - \exp(-sq) \geq 1 - \delta$.

We can immediately see that the normalizing constants for wedge and path sampling drive
up the number of samples required to get top weighted triangles, as compared to
edge sampling. However, obtaining samples with \cref{alg:edge} can be much more
costly if we have to find common neighbors of nodes with large degree. It is
not immediately clear which algorithm should be better, but our experimental
results in the next section show that edge sampling is superior in practice.

\subsection{Extensions to cliques}\label{subsec:cliquesampling}
All of our sampling algorithms can also be used to sample cliques of arbitrary size $k$. The high level objective of each sampling algorithm is to sample some sort of subgraph with probability proportional to its weight. The three sampling algorithms we present are ways to sample edges, wedges, and length-3 paths respectively. To convert one of our triangle samplers to a clique sampler, we can simply sample a subgraph with one of our three algorithms, and then enumerate all cliques incident on that subgraph. For edge sampling, we can sample an edge and enumerate all $(k-2)$-cliques incident on that edge. For wedge and path sampling, we can sample wedges and paths until we sample a triangle. Once a triangle is sampled, we can then enumerate all $(k-3)$-cliques incident on that triangle.  To test out this concept, we include additional experiments where we extend edge sampling into a clique sampler.

This natural extension to cliques is an advantage over our dynamic heavy-light algorithm
from \cref{sec:heavylight}, which is not easily extended to larger cliques. However, our main
focus is on triangles, for which we provide several numerical experiments in the next
section. We provide some experiments for cliques in \cref{app:cliques}.

\section{Numerical experiments}\label{sec:experiments}

We now report the results of several experiments measuring the performance of
our deterministic and random sampling algorithms compared to some competitive
baselines such as the static heavy-light thresholding method in
\cref{alg:staticheavylight}. Implementations of our algorithms and code to reproduce results
are available.\footnote{\url{https://github.com/raunakkmr/Retrieving-top-weighted-triangles-in-graphs}}

We find that edge sampling works much better than
wedge sampling for our random sampling algorithms but that our deterministic
heavy light algorithm is even faster across a wide range of datasets,
outperforming the baselines by orders of magnitude in terms of running time.

\begin{table}[b]
\caption{Summary statistics of datasets.}
\centering
\scalebox{0.9}{
\begin{tabular}{r c c c c}
\toprule
\multirow{2}{*}{dataset} & \multirow{2}{*}{\# nodes} & \multirow{2}{*}{\# edges} & \multicolumn{2}{c}{edge weight} \\ \cmidrule{4-5}
& & & mean & max \\
\midrule
tags-stack-overflow    & 50K  & 4.2M & 13  & 469 \\
threads-stack-overflow & 2.3M & 21M  & 1.1 & 546 \\
Wikipedia-clickstream  & 4.4M & 23M  & 347 & 817M  \\
Ethereum               & 38M  & 103M & 2.8 & 1.9M \\
AMiner                 & 93M  & 324M & 1.3 & 13K \\
reddit-reply           & 8.4M & 435M & 1.5 & 165K \\
MAG                    & 173M & 545M & 1.7 & 38K \\
Spotify                & 3.6M & 1.9B & 8.6 & 2.8M \\
\bottomrule
\end{tabular}
}
\label{tab:dataset}
\end{table}

\subsection{Data}\label{sec:data}

We used a number of datasets in order to test the performance our algorithms.
The datasets can be found at \url{http://www.cs.cornell.edu/~arb/data/index.html}.
\cref{tab:dataset} lists summary statistics of the datasets and we describe
them briefly below.

\xhdr{tags-stack-overflow~\cite{simplicial}}
On Stack Overflow, users ask, answer, and discuss computer programming
questions, and users annotate questions with 1--5 tags. The nodes in this graph
correspond to tags, are the weight between two nodes is the number of questions
jointly annotated by the two tags.

\xhdr{threads-stack-overflow~\cite{simplicial}}
This graph is constructed from a dataset of user co-participation on Stack
Overflow question threads that last at most 24 hours.  Nodes
correspond to users and the weight of an edge is the number of times that two
users appeared in one of these short question threads.

\xhdr{Wikipedia-clickstream~\cite{Wikipedia-clickstream}}
This graph is derived from Wikipedia clickstream data from request
logs in January, 2017 that capture how users transition between articles (only
transitions appearing at least 11 times were recorded).  The nodes of the graph
correspond to Wikipedia articles and the weight between two nodes is the number
of times users transitioned between the two pages.

\xhdr{Ethereum}
Ethereum is a blockchain-based computing platform for decentralized
applications. Transactions are used to update state in the Ethereum network, and
each transaction has a sender and a receiver address. We create a graph where
the nodes are addresses and the weight between two nodes is the number of
transactions between the two addresses, using all transactions on the platform
up to August 17, 2018, as provided by \texttt{blockchair.com}.

\xhdr{AMiner and MAG~\cite{Tang-2008-AMiner,MAG-data}}
These graphs are constructed by two large bibliographic databases---AMiner and
the Microsoft Academic Graph.  We construct weighted co-authorship graphs, where
nodes represent authors and the weight between two nodes is the number of papers
they have co-authored. Papers with more than 25 authors were omitted
from the graph construction.

\xhdr{reddit-reply~\cite{hessel2016science,samplingtemporal}}
Users on the social media web site \texttt{reddit.com} interact by commenting on
each other's posts. We derive a graph from a collection of user
comments. Nodes are users and the weight of an edge is the number of
interactions between the two users.

\xhdr{Spotify}
As part of a machine learning challenge, the music streaming platform Spotify
released a large number of user ``listening sessions,'' each consisting of a set
of songs. We constructed a weighted graph where the nodes represent songs and
the weight of an edge is the number of times that the songs co-appeared in a
session.

\subsection{Algorithm benchmarking}
We evaluate the performance of our proposed algorithms:
(i) random edge sampling (ES) as in \cref{alg:edge};
(ii) random wedge sampling (WS) as in \cref{alg:wedge};
(iii) random path sampling (PS) as in \cref{alg:path};
(iv) the static heavy-light (SHL) scheme as in \cref{alg:staticheavylight} (see
below for how we set the thresholds);
(v) the dynamic heavy-light scheme (DHL) as in \cref{alg:dynamicheavylight}
(vi) auto heavy-light (auto-HL), which is the distribution-obvlivious adaptation
of DHL that automatically adjusts edge promotion to optimize the decrease in the
threshold.
As a baseline, we use an optimized sequential version of
NodeIterator++~\cite{BFNPSW15, SV11, CN85} and refer to this as the ``brute
force approach'' (BF). Essentially, this algorithm iterates over vertices with
decreasing degree, and for each vertex it only enumerates triangles that are
formed by neighboring vertices with lower degree than itself.
All algorithms were implemented in C++, and all experiments were executed on a
64 core 2.20 GHz Intel Xeon CPU with 200 GB of RAM. We used parallel sorting for
all algorithms and parallel sampling for the random sampling algorithms; other
parts of the algorithms were executed sequentially. We evaluated the algorithms
for two values of $k$: 1,000 and 100,000.
We also use the arithmetic mean ($p = 1$ in \cref{eq:genmean}).

Recall that DHL (\cref{alg:dynamicheavylight}) uses a power law distribution
model of the edge weights and sets a parameter $\alpha_p$ based on the
power law exponent. We fix $\alpha_p = 1.25$ for our experiments, which
works well across a range of datasets.

The random sampling algorithms are not guaranteed to enumerate all of the
top-$k$ weighted triangles. Instead, we measure the performance of these
algorithms in terms of running time and accuracy (the fraction of top-$k$
triangles actually enumerated). We ran ES long enough for it to achieve at least
94\% ($k =$ 1,000) or 50\% ($k =$ 100,000) accuracy on all datasets.
We also ran WS long enough for it to achieve at least 50\% accuracy (for both
values of $k$). However, in practice, its performance is poor, and we terminate
the algorithm if it takes longer than BF to achieve this accuracy level. PS takes longer than BF to achieve this accuracy level on all datasets, so we do not include it in \cref{tab:compare_results}.

Similarly, the static heavy-light algorithm (\cref{alg:staticheavylight}) is not
guaranteed to achieve 100\% accuracy since it relies on a fixed threshold to
partition the edges as heavy and light and only enumerates triangles formed by
heavy edges. As the threshold decreases a larger number of edges are labelled
heavy. This increases the accuracy but also slows down the algorithm.
\Cref{fig:static_hl_tradeoff} illustrates this trade-off
on the Ethereum dataset; a similar trend is observed on the
other datasets. In our experiments, we labelled the top 10\% of edges as heavy
and report the achieved accuracy. As we discuss below, we find that SHL is
slower and attains sub-100\% accuracy in practice; improving the accuracy would
only make this baseline slower.

\begin{figure}[tb]
  \centering
  \includegraphics[width=0.495\columnwidth]{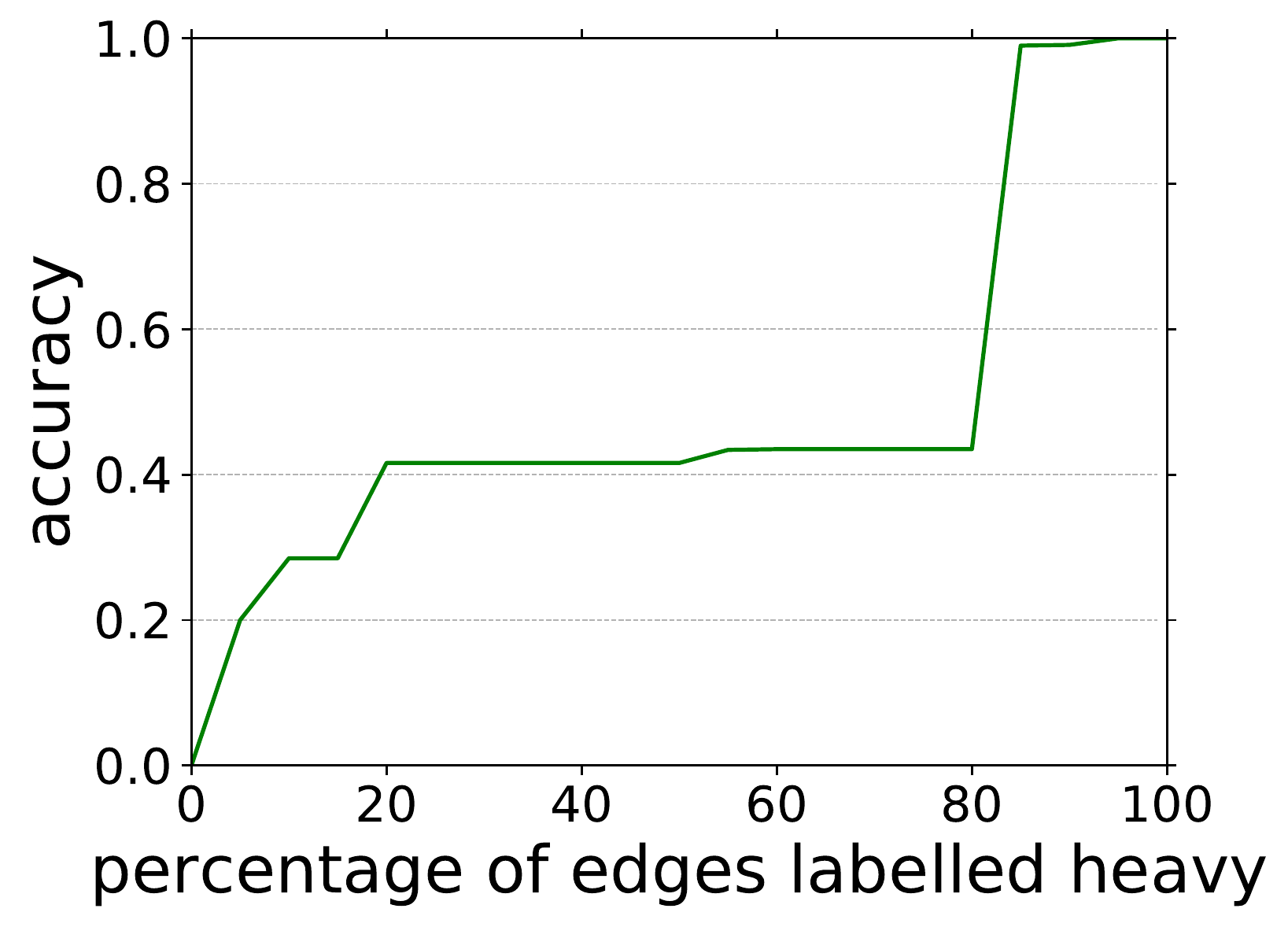}
  \includegraphics[width=0.495\columnwidth]{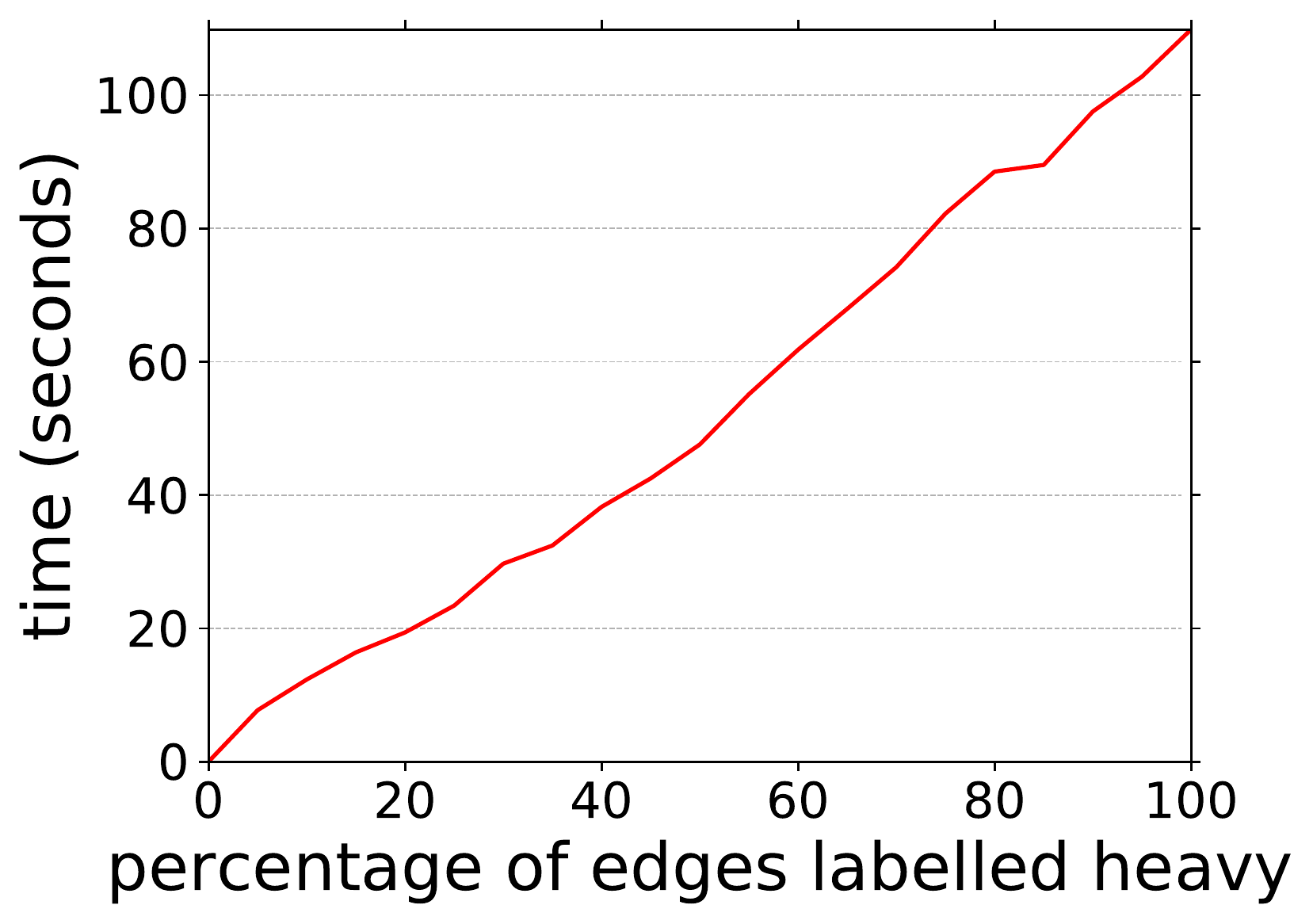}
  \caption{Accuracy and running time as a function of edges labeled ``heavy'' by
    the thresholding for the static heavy-light algorithm
    (\cref{alg:staticheavylight}) on the Ethereum dataset for $k = $ 1000. As
    the threshold decreases, a larger percentage of edges are labelled
    heavy. This increases the accuracy but also increases the running time.
    For reasonable accuracy levels, we find that the running time is slower
    than our optimized dynamic heavy-light algorithm (\cref{alg:dynamicheavylight}),
    which achieves 100\% accuracy (see \cref{tab:compare_results}).
  }
\label{fig:static_hl_tradeoff}
\end{figure}

\begin{table*}[tb]
\caption{Running times of all of our algorithms in seconds averaged over 10 runs. 
BF is brute force enumeration of triangles, which is the out-of-the-box baseline; 
ES is the parallel edge sampling algorithm (\cref{alg:edge});
WS is the parallel wedge sampling algorithm (\cref{alg:wedge});
SHL is the static heavy-light threshold deterministic algorithm (\cref{alg:staticheavylight});
DHL is the dynamic heavy-light deterministic algorithm (\cref{alg:dynamicheavylight}); and
Auto-HL is the distribution oblivious modification of the dynamic heavy-light discussed in \cref{sec:analysis}.
ES was run to achieve 94\% ($k =$ 1,000) or 50\% ($k =$ 100,000) accuracy, while
WS was run to achieve just 50\% accuracy and was stopped early if taking longer BF
(or longer than SHL on the Spotify dataset).
SHL is also an approximation, and we report its  accuracy in the final column.
Overall, our deterministic algorithms (DHL or Auto-HL) are fast and achieve
100\% accuracy; in some cases, our ES technique is slightly faster, but it is approximate.
}
\label{tab:compare_results}
\begin{tabular}{r r r r r r r r c}
 \toprule
$k$ & dataset & BF & ES & WS & DHL & Auto-HL & SHL & Accuracy (SHL)\\
 \midrule
\multirow{3}{*}{1000} 
 & tags-stack-overflow    & 12.71    & 0.57          & 11.28     & \textbf{.08}  & 0.09           & 0.54    & 0.99 \\
 & threads-stack-overflow & 34.92    & 1.31          & $>$34.92  & 0.53          & \textbf{0.38}  & 1.55    & 0.55 \\
 & Wikipedia-clickstream  & 16.32    & 14.31         & $>$16.32  & \textbf{5.44} & 7.26           & 2.02    & 0.87 \\
 & Ethereum               & 52.91    & 9.03          & $>$52.91  & 8.12          & \textbf{6.94}  & 11.90   & 0.28 \\
 & Aminer                 & 243.75   & \textbf{3.72} & $>$243.75 & 13.35         & 12.36          & 43.47   & 0.32 \\
 & reddit-reply           & 4047.62  & 5.19          & 341.17    & 5.02          & \textbf{4.74}  & 102.65  & 0.51 \\
 & MAG                    & 512.24   & \textbf{4.92} & 48.58     & 29.19         & 20.89          & 72.49   & 0.91 \\
 & Spotify                & $>$86400 & 60.33         & $>$5300   & 31.82         & \textbf{30.79} & 5388.45 & 1.00 \\
 \midrule
\multirow{3}{*}{100000} 
 & tags-stack-overflow    & 13.06    & 0.58          & $>$13.06   & \textbf{0.23}  & \textbf{0.23}  & 0.62    & 0.28 \\
 & threads-stack-overflow & 33.99    & \textbf{1.19} & $>$33.99   & 1.82           & 1.73           & 1.63    & 0.32 \\
 & Wikipedia-clickstream  & 17.34    & 13.64         & $>$17.34   & \textbf{5.49}  & 7.24           & 2.15    & 0.13 \\
 & Ethereum               & 57.35    & 10.03         & $>$57.35   & \textbf{18.11} & 19.87          & 11.70   & 0.11 \\
 & Aminer                 & 245.28   & 3.45          & $>$245.28  & 15.38          & \textbf{13.91} & 43.28   & 0.24 \\
 & reddit-reply           & 3857.57  & \textbf{6.52} & $>$3857.57 & 6.87           & 7.49           & 98.34   & 0.34 \\
 & MAG                    & 524.80   & \textbf{4.25} & $>$524.80  & 29.52          & 21.37          & 75.97   & 0.10 \\
 & Spotify                & $>$86400 & 47.27         & $>$5300    & 30.57          & \textbf{29.89} & 5384.17 & 0.92 \\
 \bottomrule
\end{tabular}
\end{table*}

\Cref{tab:compare_results} shows the running times of all of our algorithms.  BF
did not terminate on Spotify after 24 hours, so running times for this baseline
are not available on that dataset.
We highlight a few important findings. First, our deterministic methods DHL and
Auto-HL excel at retrieving the top $k$ triangles. They achieve perfect accuracy
with orders of magnitude faster than BF. For instance, these algorithms get a
1000x speedup on the reddit-reply dataset ($k = 1,000$) and more than a 2000x
speedup on the Spotify dataset ($k = 100,000$). These algorithms also
outperform SHL by a significant margin in terms of time and accuracy. For
example, despite being 30x slower on reddit-reply SHL only achieves 50\%
accuracy ($k = 1,000$). Again, our deterministic algorithms algorithms
\emph{always achieve 100\% accuracy}, and do so in a fraction of the time taken
by the baseline methods BF and SHL.

Among the sampling algorithms, ES performs much better than WS. WS struggles to
achieve high accuracy and is not competitive with the BF baseline or SHL. On the other
hand, ES is quite competitive with even DHL and Auto-HL. ES retrieves the top
1000 triangles on the AMiner and MAG datasets with 99\% accuracy at speedups of
2x or 4x over DHL and Auto-HL. A similar speedup is observed for $k=$ 100,000,
but ES is only achieving 50\% accuracy in these cases. Even though ES works well
in these cases, our deterministic algorithms are still competitive; we conclude
that intelligent deterministic approaches work extremely well for finding top
weighted triangles in large weighted graphs.
  
All of our algorithms except BF and WS use a pre-procesing step of sorting edges
by weight. Surprisingly, we find that this pre-processing step is the bottleneck
in our computations. Sorting in parallel is crucial to achieving high
performance. In turn, this negates the possible benefit of parallel sampling for
the randomized algorithms over our deterministic methods, whose main routines
are inherently sequential.

\hide{
Additionally, the DHL and Auto-HL algorithms rely on manipulation of the
subgraph induced by light edges. A simple way to allow for these operations is
to use a map data structure to store the neighbors (and the corresponding edge
weights) of every vertex. However, building up this data structure is extremely
slow in practice due to large constant factors and massive datasets. On the
other hand, using lazy deletions or computing the map data structure using
memoization makes the preprocessing upto 300x faster at the cost of making the
main loop of DHL and Auto-HL 2-3x slower.}

\section{Discussion}\label{sec:discussion}
Subgraph patterns, and in particular, triangles, have been used extensively in
graph mining applications. However, most of the existing literature focuses on
counting or enumeration tasks in unweighted graphs. In this paper, we developed
deterministic and random sampling algorithms for finding the heaviest triangles
in large weighted graphs. With some tuning, our main deterministic algorithm can
find these triangles in a few seconds on graphs with hundreds of millions of
edges or in 30 seconds on a graph with billions of edges. This is orders of
magnitude faster than what one could achieve with existing fast enumeration
schemes and is usually much faster than even our randomized sampling algorithms.

We anticipate that our work will enable scientists to better explore large-scale
weighted graphs and also expect that our work will spur new algorithmic
developments on subgraph listing and counting in weighted graphs. For example,
an interesting avenue for future research would be the development of random
sampling algorithms that sample triangles with probability proportional to some
arbitrary function of their weight, chosen to converge to
the top weighted triangles faster. This could make random sampling approaches
competitive with our fast deterministic methods. The edge sampling method can
also be generalized to $k$-clique sampling by sampling an edge and then
enumerating adjacent $k$-cliques. How to extend the deterministic algorithms to
top $k$-clique enumeration is less clear, so sampling may be more appropriate
for larger clique patterns.

\section*{Acknowledgements}
This research was supported in part by
NSF Award DMS-1830274,
ARO Award W911NF19-1-0057,
ARO MURI,
NSF grant CCF-1617577, Simons Investigator Award, Google Faculty Research Award, Amazon Research Award, and
Google Cloud resources.




\appendix
\label{app:cliques}

\section{Additional Numerical Experiments}

While our methods focused on finding top weighted triangles, some of the
sampling methods naturally extend to the larger cliques. In this case, the
weight of a clique is some generalized $p$-mean of the weights of the
edges in the clique.
We found that the extension of edge sampling---as described in \cref{subsec:cliquesampling}---performed best
in practice, and we compare the performance of edge sampling algorithm to an
intelligent brute force approach for finding top weighted 4-cliques and 5-cliques
that enumerates all cliques using the algorithm of Chiba and Nishizeki~\cite{CN85}.
Our main finding is that edge sampling can approximately retrieve the top weighted cliques
much faster; however the performance does show higher variance than for the case of triangles.


Since brute force enumeration of even 4-cliques is computationally expensive, we
use smaller datasets for these experiments than the ones in \cref{tab:dataset}. 
We construct weighted graphs from 5 temporal hypergraph datasets~\cite{simplicial},
where the weight of an edge $(u,v)$ is the number of hyperedges that contain nodes
$u$ and $v$. \Cref{tab:clique_dataset} shows summary statistics for the data.

\begin{table}[b]
\caption{Summary statistics of datasets used in clique sampling experiments.}
\centering
\begin{tabular}{r c c c c}
\toprule
\multirow{2}{*}{dataset} & \multirow{2}{*}{\# nodes} & \multirow{2}{*}{\# edges} & \multicolumn{2}{c}{edge weight} \\ \cmidrule{4-5}
& & & mean & max \\
\midrule
email-Enron            & 143  & 1800 & 16  & 819 \\
email-Eu               & 979  & 29K  & 25  & 48K \\
contact-high-school    & 327  & 5.8K & 32  & 29K \\
contact-primary-school & 242  & 8.3K & 15  & 780 \\
tags-math-sx           & 1.6K & 91K  & 17  & 16K \\
\bottomrule
\end{tabular}
\label{tab:clique_dataset}
\end{table}


We evaluate the performance of our proposed algorithm,
random edge sampling (ES) as in \cref{alg:edge} with the modifications
mentioned in \cref{subsec:cliquesampling}. As a baseline, we use an
optimized sequential clique enumerator and refer to this as the
``brute force approach'' (BF).
The rest of the experimental setup is the same as in \cref{sec:experiments}.
We evaluated the algorithm for two values of $k$: $1,000$ and $100,000$,
and two clique sizes: 4 and 5.
We also use the arithmetic mean ($p = 1$ in \cref{eq:pmean_unnorm}).

The random edge sampling
algorithm is not guaranteed to enumerate all of the top-$k$
weighted cliques. Instead, we measure the performance of this algorithm
in terms of running time and accuracy (the fraction of top-$k$ cliques actually
enumerated). We ran ES long enough for it
to achieve at least 50\% (clique size = 4) or at least 40\% (clique size = 5) accuracy on all datasets.

\cref{tab:clique_exp} shows the running  times of BF and ES. We observe substantial speedup of ES on all datasets---for example, we see speedups of 14x on the email-Eu dataset for both 4-cliques and 5-cliques.
However, ES is only required to achieve 40 to 50\% accuracy in these cases.
We find that, unlike the performance of ES for finding top weighted triangles, 
the performance of ES for finding top weighted cliques has higher variance across runs.
In some cases, ES can actually become more expensive than the brute force approach (5-cliques on the tags-math-sx dataset);
a better understanding of these changes in performance is still an open question.

\begin{table}[hb]
\caption{
Running times in seconds for brute force enumeration (BF) and parallel
edge sampling (ES) (averaged over 10 runs) for 4 and 5 cliques. ES was run long enough to achieve 50\% (4-cliques)
and 40\% (5-cliques) accuracy for $k = 1,000$ and $100,000$.
}
\begin{tabular}{r r r r r r r r}
 \toprule
    &         & \multicolumn{4}{c}{clique size} \\
    \cmidrule{3-6}
    &         & \multicolumn{2}{c}{4} & \multicolumn{2}{c}{5} \\
 \midrule
$k$ & dataset & BF & ES & BF & ES \\
 \midrule
\multirow{3}{*}{1000} 
 & email-Enron            &  1.13  & 0.05 & 1.1    & 0.3    \\
 & email-Eu               &  9.17  & 0.4  & 83     & 5.8    \\
 & contact-high-school    &  1.32  & 0.08 & 1.24   & 0.5    \\
 & contact-primary-school &  1.63  & 0.3  & 9      & 2.5    \\
 & tags-math-sx           &  398   & 4.5  & 9340   & $>$9340       \\
 \midrule
\multirow{3}{*}{100000} 
 & email-Enron            &  1.13  & 0.05 & 1.1    & 0.3   \\
 & email-Eu               &  9.17  & 0.45 & 83     & 5.8   \\
 & contact-high-school    &  1.32  & 0.08 & 1.24   & 0.5   \\
 & contact-primary-school &  1.63  & 0.2  & 9      & 2.5   \\
 & tags-math-sx           &  398   & 4.5  & 9340   & $>$9340      \\
 \bottomrule
\end{tabular}
\label{tab:clique_exp}
\end{table}{}


\end{document}